\documentclass[12pt]{article}

\usepackage{url}

\usepackage{subcaption}

\usepackage{graphicx}
\usepackage{multirow}

\usepackage{amssymb,amsmath,amsthm,amsfonts,pifont} 

\usepackage[inline]{enumitem}

\usepackage{tikz}
\usetikzlibrary{arrows,shapes,automata,petri,positioning}

\tikzset{
	place/.style={
		circle,
		thick,
		draw=blue!75,
		fill=blue!20,
		minimum size=6mm,
	},
	transitionH/.style={
		rectangle,
		thick,
		fill=black,
		minimum width=8mm,
		inner ysep=2pt
	},
	transitionV/.style={
		rectangle,
		thick,
		fill=black,
		minimum height=8mm,
		inner xsep=2pt
	}
}

\usepackage{forest}
\usepackage{tikz-qtree}

\newcommand{\dng}[1]{{\sim}{#1}}

\newcommand{\dg}[1]{\text{dg}({#1})}
\newcommand{\pdg}[1]{\text{dg}^+({#1})}
\newcommand{\ig}[1]{\text{ig}({#1})}
\newcommand{\fix}[1]{\text{fix}({#1})}
\newcommand{\atom}[1]{\text{atom}({#1})}
\newcommand{\head}[1]{\text{h}({#1})}
\newcommand{\body}[1]{\text{b}({#1})}
\newcommand{\bodyf}[1]{\text{bf}({#1})}
\newcommand{\pbody}[1]{\text{b}^+({#1})}
\newcommand{\nbody}[1]{\text{b}^-({#1})}
\newcommand{\sm}[1]{\text{sm}({#1})}
\newcommand{\supp}[1]{\text{supp}({#1})}
\newcommand{\cf}[1]{\text{cf}({#1})}
\newcommand{\var}[1]{\text{var}_{{#1}}}

\newtheorem{theorem}{Theorem}
\newtheorem{lemma}[theorem]{Lemma}

\newtheorem{proposition}[theorem]{Proposition}

\newtheorem{definition}{Definition}
\newtheorem{example}{Example}

\usepackage[margin=1in]{geometry}

\begin{document}

\title{\bf Static Analysis of Logic Programs via Boolean Networks}

\author{Van-Giang Trinh\thanks{LIRICA team, LIS, Aix-Marseille University, Marseille, France; email: \texttt{trinh.van-giang@lis-lab.fr}}
  \and
  Belaid Benhamou\thanks{LIRICA team, LIS, Aix-Marseille University, Marseille, France; email: \texttt{belaid.benhamou@lis-lab.fr}}
  }

\maketitle

\begin{abstract}
Answer Set Programming (ASP) is a declarative problem solving paradigm that can be used to encode a combinatorial problem as a logic program whose stable models correspond to the solutions of the considered problem.
ASP has been widely applied to various domains in AI and beyond.
The question ``What can be said about stable models of a logic program from its static information?'' has been investigated and proved useful in many circumstances.
In this work, we dive into this direction more deeply by making the connection between a logic program and a Boolean network, which is a prominent modeling framework with applications to various areas.
The proposed connection can bring the existing results in the rich history on static analysis of Boolean networks to explore and prove more theoretical results on ASP, making it become a unified and powerful tool to further study the static analysis of ASP.
In particular, the newly obtained insights have the potential to benefit many problems in the field of ASP.
\end{abstract}

\section{Introduction}
\label{sec:Introduction}

Answer Set Programming (ASP) has emerged as a powerful declarative programming paradigm for solving complex combinatorial problems under the umbrella of logic programming and non-monotonic reasoning~\cite{gelfond1988stable}.
Its basic idea is to describe the specifications of a problem by means of a logic program: solutions of this problem will be represented by the stable models of the logic program~\cite{eiter2009answer}.
Then the logic program can be fed into a software system called an ASP solver such as DLV, Cmodels, Smodels, or clasp~\cite{DBLP:journals/aicom/GebserKKOSS11}, which can compute some or all stable models of this program.
With the constant improvements in both theory (e.g., expressive powers) and practice (e.g., efficient ASP solvers), ASP has been successfully applied to a wide range of areas in AI (e.g., planning, knowledge representation and reasoning, semantic web, natural language processing, and argumentation) and beyond (e.g.,  systems biology, computational biology, bounded model checking, software engineering, robotics, and manufacturing)~\cite{erdem2016applications}.

In general, the problem of deciding if a given ground normal logic program has some stable model is NP-complete~\cite{DBLP:journals/jacm/MarekT91}.
It is thus natural to consider the static analysis of ASP, i.e., answering the question: \emph{What can be said about stable models of a logic program from its static information?}
Herein, a Logic Program (LP) is considered as a ground normal logic program~\cite{gelfond1988stable} and its static information is mostly considered as graph representations of the LP such as (positive) dependence graphs~\cite{cois1994consistency}, cycle and extended dependence graphs~\cite{costantini2006existence,DBLP:conf/gkr/CostantiniP11}, rule graphs~\cite{DBLP:journals/tcs/DimopoulosT96}, and block graphs~\cite{DBLP:conf/ijcai/Linke01}.
Historically, the first studies of this research direction focused on the existence of a unique stable model in dependence graph-based classes of LPs including positive LPs~\cite{gelfond1988stable}, acyclic LPs~\cite{apt1991acyclic}, and locally stratified LPs~\cite{gelfond1988stable}.
In 1994, Fages proved the most important result about the coincidence between stable models and models of the Clark's completion (equivalently supported models) in tight LPs~\cite{cois1994consistency}.
Being finer-represented than dependence graphs, cycle and extended dependence graphs were introduced and several improved results were obtained~\cite{costantini2006existence,DBLP:conf/gkr/CostantiniP11}.
However, there is quite limited progress regarding this direction.
Moreover, it is noted that the cycle graph or the extended dependence graph of an LP can have exponentially many vertices as the LP can have an exponential number of cycles~\cite{costantini2006existence} or rules~\cite{DBLP:conf/gkr/CostantiniP11}, respectively.
Recently, we have been only aware of the work~\cite{FANDINNO_LIFSCHITZ_2023}, which presents new understanding of the positive dependence graph to possibly strengthen some existing theoretical results.
In summary, while the static analysis of ASP is important and has a wide range of useful applications in both theory and practice of ASP~\cite{cois1994consistency,costantini2006existence,DBLP:conf/kr/EiterHK21}, there is quite limited progress on this research topic.

Boolean Networks (BNs) are a simple yet efficient mathematical formalism that has been widely applied to many areas from science to engineering~\cite{schwab2020concepts,Paulev2020}: mathematics, systems biology, computer science, theoretical physics, robotics, social modeling, neural networks, etc.
A BN includes \(n\) Boolean \emph{variables} associated with \(n\) corresponding Boolean \emph{functions} that determine the value update of a variable over time.
In this sense, a BN can be seen as a discrete dynamical system~\cite{thomas1973boolean,thomas1990biological}.
Originated by the early work~\cite{thomas1990biological}, the static analysis of BNs (i.e., studying relations between the dynamics of a BN and its influence graph) has a rich history of research~\cite{richard2019positive}.
In particular, most results focus on exploring relations between fixed points and the influence graph of a BN, where a fixed point is a state once entered, the BN's dynamics cannot escape from it~\cite{aracena2008maximum}.
To date, this research direction is still growing with many prominent and deep results~\cite{richard2019positive,richard2018fixed,richard2023attractor}.

ASP has found widespread applications in both the modeling and analysis of BNs.
It has been notably used for important tasks on BNs (all are computational ones) such as fixed point enumeration~\cite{Paulev2020}, enumeration and approximation of attractors~\cite{Paulev2020,khaled2023using}, inference of BNs from biological data~\cite{Ribeiro2021}, model repair~\cite{DBLP:journals/tplp/GebserSTV11}, and reprogramming~\cite{DBLP:journals/tplp/KaminskiSSV13}.
In the opposite direction, there is however no significant work trying to use BNs for studying the static analysis of ASP.
Indeed, the initial connection between ASP and BNs can be traced back to the theoretical work by~\cite{inoue2011logic}.
In this work, the authors used LPs to model BNs, and vice versa.
However, their investigations were limited to the conceptual level, and no subsequent studies have delved deeper into this connection so far.

In this work, we try to explore in depth the static analysis of ASP.
Herein, static information of an LP is considered as its Dependence Graph (DG)~\cite{cois1994consistency}.
The DG can be efficiently (syntactically) computed from the LP itself and might be in general more compact than other graph representations (e.g., cycle or extended dependence graphs~\cite{costantini2006existence,DBLP:conf/gkr/CostantiniP11}), but can tell many significant insights about stable models of the LP as we shall show in this paper.
Our main contribution is to consolidate the initial connection between LPs and BNs proposed in~\cite{inoue2011logic}, and in particular to use this connection for the first time to explore and prove several relations between stable models of an LP and its DG.
More specifically, we obtain several new theoretical results on how positive and negative cycles on the DG affect the number of stable models of the LP.
Furthermore, we also demonstrate several potential applications of the newly obtained results to other important problems in the field of ASP.
All the above elements show that the approach proposed in the present paper can bring the plenty of results in the field of BNs (and the field of discrete dynamical systems in general) to the field of ASP as well as provide a unified framework for exploring and proving more new theoretical results in ASP.

In the rest of this paper, Section~\ref{sec:Preliminaries} recalls preliminaries on logic programs and Boolean networks. 
Section~\ref{sec:Main-Results} presents the connection between logic programs and Boolean networks along with the relations between stable models and positive and/or negative cycles.
Section~\ref{sec:Discussion} presents discussions on potential applications and extensions of the new theoretical results.
Section~\ref{sec:Conclusion} concludes the paper and draws future work.

\section{Preliminaries}\label{sec:Preliminaries}

In this paper, we consider only ground normal logic programs.
Unless specifically stated, a Logic Program (LP) means a ground normal logic program.
\(\mathbb{B} = \{0, 1\}\) is the Boolean domain, and all Boolean operators used in this paper include \(\land\) (conjunction), \(\lor\) (disjunction), \(\neg\) (negation), and \(\leftrightarrow\) (bi-implication).

An LP \(P\) is a finite set of rules of the form \(p \gets p_1 \land \dots \land p_m \land \dng{p_{m + 1}} \land \dots \land \dng{p_{k}}\) where \(p\) and \(p_i\) are variable-free atoms (\(k \geq m \geq 0\)), \(\sim\) denotes the negation as failure.
We use \(\atom{P}\) to denote the set of all atoms of \(P\).
For any rule \(r\) of this form, \(p\) is called the \emph{head} of \(r\) (denoted by \(\head{r}\)), \(\pbody{r} = \{p_1, \dots, p_m\}\) is called the \emph{positive body} of \(r\), \(\nbody{r} = \{p_{m + 1}, \dots, p_{k}\}\) is called the \emph{negative body} of \(r\), \(\body{r} = \pbody{r} \cup \nbody{r}\) is the \emph{body} of \(r\), and \(\bodyf{r} = p_1 \land \dots \land p_m \land \neg p_{m + 1} \land \dots \land \neg p_{k}\) is the \emph{body formula} of \(r\).
If \(\body{r} = \emptyset\), then \(r\) is called a \emph{fact} and \(\bodyf{r}\) is conventionally 1.
An \emph{Herbrand interpretation} \(I\) is a subset of \(\atom{P}\), and is called an \emph{Herbrand model} if for any rule \(r\) in \(P\), \(\pbody{r} \subseteq I\) and \(\nbody{r} \cap I = \emptyset\) imply \(\head{r} \in I\).
\(P\) is \emph{positive} if \(\nbody{r} = \emptyset\) for all \(r \in P\).
In this case, \(P\) has a unique least Herbrand model.
An Herbrand interpretation \(I\) is a \emph{stable model} of \(P\) if \(I\) is the least Herbrand model of the reduct of \(P\) with respect to \(I\) (denoted by \(P^I\)) where \(P^I = \{\head{r} \leftarrow \bigwedge_{p \in \pbody{r}} p \vert r \in P, \nbody{r} \cap I = \emptyset\}\).
We use \(\sm{P}\) to denote the set of all stable models of \(P\).

The Clark's completion of an LP \(P\) (denoted by \(\cf{P}\)) consists of the following sentences: for each \(p \in \atom{P}\), let \(r_1, \dots, r_l\) be all the rules of \(P\) having the same head \(p\), then \(p \leftrightarrow \bodyf{r_1} \lor \dots \lor \bodyf{r_l}\) is in \(\cf{P}\). If \(l = 0\), then the equivalence is \(p \leftrightarrow 0\).
In this paper, we shall identify a truth assignment with the set of atoms true in this assignment, and conversely, identify a set of atoms with the truth assignment that assigns an atom true iff it is in the set.
Under this convention, a model of \(\cf{P}\) is also a Herbrand model of \(P\).
An Herbrand model \(A\) is a \emph{supported model} if for any atom \(p \in A\), there is a rule \(r \in P\) such that \(\head{r} = p\), \(\pbody{r} \subseteq A\), and \(\nbody{r} \cap A = \emptyset\).
We use \(\supp{P}\) to denote the set of all supported models of \(P\).
It is well-known that supported models of \(P\) coincide with models of its Clark's completion~\cite{apt1988towards}.
In particular, a stable model is also a supported model, but the converse is not true in general.

We define the Dependence Graph (DG) of \(P\) (denoted by \(\dg{P}\)) as a signed directed graph \((V, E)\) on the set of signs \(\{\oplus, \ominus\}\) where \(V = \atom{P}\) and \((uv, \oplus) \in E\) (resp. \((uv, \ominus) \in E\)) iff there is a rule \(r \in P\) such that \(v = \head{r}\) and \(u \in \pbody{r}\) (resp. \(u \in \nbody{r}\)).
An arc \((uv, \oplus)\) is positive, whereas an arc \((uv, \ominus)\) is negative.
The positive DG of \(P\) (denoted by \(\pdg{P}\)) is a subgraph of \(\dg{P}\) that has the same set of vertices but contains only positive arcs.

A Boolean Network (BN) \(f\) is a set of Boolean functions on a set of Boolean variables denoted by \(\var{f}\).
Each variable \(v\) is associated with a Boolean function \(f_v \colon \mathbb{B}^{|\var{f}|} \rightarrow \mathbb{B}\).
\(f_v\) is called \emph{constant} if it is always either 0 or 1 regardless of its arguments.
A state \(s\) of \(f\) is a mapping \(s \colon \var{f} \mapsto \mathbb{B}\) that assigns either 0 (inactive) or 1 (active) to each variable.
We can write \(s_v\) instead of \(s(v)\) for short.
At each time step, all variables are updated simultaneously, i.e., \(s'_v = f_v(s), \forall v \in \var{f}\), where \(s\) (resp. \(s'\)) is the state of \(f\) at time \(t\) (resp. \(t + 1\)).
\(s'\) is called the successor state of \(s\).
By abuse of notation, we can consider a state \(x\) as a subset of \(\var{f}\) where \(v \in x\) if and only if \(x_v = 1\).
\(x\) is said to be a \emph{fixed point} of \(f\) if \(f_v(x) = x_v, \forall v \in \var{f}\).
We use \(\fix{f}\) to denote the set of all fixed points of \(f\).

Let \(G = (V, E)\) be a signed directed graph on the set of signs \(\{\oplus, \ominus\}\).
The \emph{in-degree} of a node \(v\) is the number of vertices having directed arcs to \(v\).
The minimum in-degree of \(G\) is the smallest in-degree of its vertices.
A cycle \(C\) of \(G\) is positive (resp. negative) if it contains an even (resp. odd) number of negative arcs.
\(G\) is \emph{sign-definite} if there cannot be two arcs with different signs between two different nodes.
\(G\) is \emph{strongly connected} if there is always a path between two any vertices of \(G\).
Let \(x\) be a state of \(f\).
We use \(x[v \leftarrow a]\) to denote the state \(y\) so that \(y_v = a\) and \(y_u = x_u, \forall u \in \var{f}, u \neq v\) where \(a \in \mathbb{B}\).
We define the Influence Graph (IG) of \(f\) (denoted by \(\ig{f} \)) as a signed directed graph \((V, E)\) on the set of signs \(\{\oplus, \ominus\}\) where \(V = \var{f}\), \((uv, \oplus) \in E\) iff there is a state \(x\) such that \(f_v(x[u \leftarrow 0]) < f_v(x[u \leftarrow 1])\), and \((uv, \ominus) \in E\) iff there is a state \(x\) such that \(f_v(x[u \leftarrow 0]) > f_v(x[u \leftarrow 1])\).

\section{Main Results}
\label{sec:Main-Results}

For clarification, we make a summary of the results that shall be presented in this section.
First, the known results with formal proofs include Theorem~\ref{theo:fages}, Proposition~\ref{proposition:acyclic-program}, and Theorem~\ref{theorem:locally-stratified-lp}.
Second, the known results in the ASP folklore without formal proofs include Theorems~\ref{theorem:no-positive-cycle} and~\ref{theorem:neg-model-existence}.
Third, the completely new results include Theorems~\ref{theorem:ig_dg} and~\ref{theorem:supported-fixed-point}, Propositions~\ref{proposition:positive-cycle} and~\ref{proposition:negative-cycle}, Lemmas~\ref{lemma:distinct-stable-model-positive-cycle} and~\ref{lemma:sign-definite-graph-scc}, Proposition~\ref{proposition:ubound-pfvs}, Lemma~\ref{lemma:fixpoint-neg}, and Theorem~\ref{theorem:one-scc-two-stable-models}.

\subsection{Logic Programs and Boolean Networks}

First of all, we recall Fages' theorem (see Theorem~\ref{theo:fages}).

\begin{theorem}[\cite{cois1994consistency}]
	Let \(P\) be an LP.
	If \(\pdg{P}\) has no cycle, then the set of stable models of \(P\) coincides with the set of supported models of \(P\).
	\label{theo:fages}
\end{theorem}

Next, we define the BN encoding for LPs (see Definition~\ref{definition:BN-encoding}).

\begin{definition}
	\label{definition:BN-encoding}
	Let \(P\) be an LP. We define a BN \(f\) encoding \(P\) as: \(\var{f} = \atom{P}\) and \(f_v = \bigvee_{r \in P, v = \head{r}}\bodyf{r}, \forall v \in \text{var}_f\).
	Conventionally, if \(v\) does not appear in the head of any rule then \(f_v = 0\).
\end{definition}

The first connection between an LP and its BN encoding is about the relation between the DG and the IG (see Theorem~\ref{theorem:ig_dg}).
This relation is very important because all the following results rely on it and the DG or the positive DG of \(P\) can be efficiently built based on the syntax only, whereas the construction of the IG of \(f\) may be exponential in general.
Note however that the IG of \(f\) is usually built by using Binary Decision Diagrams (BDDs)~\cite{richard2019positive}.
In this case,  the IG can be efficiently obtained because each Boolean function \(f_v\) is already in Disjunctive Normal Form (DNF), thus the BDD of this function would be not too large.

\begin{theorem}
	Let \(P\) be an LP and \(f\) be its BN encoding.
	Then \(\ig{f} \subseteq \dg{P}\).
	\label{theorem:ig_dg}
\end{theorem}
\begin{proof}
	We have that \(\ig{f}\) and \(\dg{P}\) have the same set of vertices that is the set of ground atoms of \(P\).
	Let \((uv, \epsilon)\) be an arc in \(\ig{f}\).
	We show that \((uv, \epsilon)\) is also an arc in \(\dg{P}\).
	Without loss of generality, suppose that \(\epsilon = \oplus\).
	
	Assume that \((uv, \oplus)\) is not an arc in \(\dg{P}\).
	There are two cases.
	Case 1: there is no arc from \(u\) to \(v\) in \(\dg{P}\).
	In this case, both \(u\) and \(\neg u\) clearly do not appear in \(f_v\).
	This implies that \(\ig{f}\) has no arc from \(u\) to \(v\), which is a contradiction.
	Case 2: there is only a negative arc from \(u\) to \(v\) in \(\dg{P}\).
	It follows that \(\neg u\) appears in \(f_v\) but \(u\) does not because \(f_v\) is in DNF.
	Then, for any state \(x\) and for every conjunction \(c\) of \(f_v\), we have that \(c(x[u \leftarrow 0]) \geq c(x[u \leftarrow 1])\).
	This implies that \(f_v(x[u \leftarrow 0]) \geq f_v(x[u \leftarrow 1])\) for any state \(x\).
	Since \((uv, \oplus)\) is an arc in \(\ig{f}\), there is a state \(x\) such that \(f_v(x[u \leftarrow 0]) < f_v(x[u \leftarrow 1])\).
	This leads to a contradiction.
	Hence, \((uv, \oplus)\) is an arc in \(\dg{P}\).
	
	We can conclude that \(\ig{f} \subseteq \dg{P}\) (i.e., \(\ig{f}\) is a subgraph of \(\dg{P}\)).
\end{proof}

\begin{theorem}
	\label{theorem:supported-fixed-point}
	Let \(P\) be an LP and \(f\) be its BN encoding.
	Then the set of supported models of \(P\) coincides with the set of fixed points of \(f\), i.e., \(\supp{P} = \fix{f}\).
\end{theorem}
\begin{proof}
	By the definition of the BN encoding and the fixed point characterization, \(\fix{f}\) is identical to the set of models of \(\text{cf}(P)\).
	The set of models of \(\text{cf}(P)\) is exactly the set of supported models of \(P\) (i.e., \(\supp{P}\))~\cite{apt1988towards}.
	Hence, \(\supp{P} = \fix{f}\).
\end{proof}

Theorem~\ref{theorem:supported-fixed-point} shows that for an LP \(P\), \(\supp{P} = \fix{f}\) where \(f\) be the encoded BN of \(P\).
It suggests that the set of fixed points of \(f\) is an upper bound for the set of stable models of \(P\).
See Example~\ref{example:stable-supported} for illustration.
If \(\pdg{P}\) has no cycle, then the two sets are the same by Fages' theorem.
In this case, all existing (both theoretical and practical) results on fixed points of BNs can be directly applied to stable models of LPs.
However, such a class of LPs is only a small piece of all possible LPs.
Hence, we shall exploit Theorem~\ref{theorem:ig_dg}, Theorem~\ref{theorem:supported-fixed-point}, and the existing results on fixed points of BNs to explore new results on stable models of LPs in general.

\begin{example} Consider an LP \(P\): \(a \leftarrow b; a \leftarrow \dng{b}; b \leftarrow c; c \leftarrow b\).
	We use ';' to separate program rules.
	The encoded BN \(f\) of \(P\) is: \(f_a = b \lor \neg b = 1, f_b = c, f_c = b\).
	Figures~\ref{fig:dg-ig-example}(a) and~\ref{fig:dg-ig-example}(b) show the DG of \(P\) and the IG of \(f\), respectively.
	We can see that \(\ig{f} \subset \dg{P}\).
	\(f\) has two fixed points: 100 and 111.
	However, only 100 is the unique stable model of \(P\) (i.e., \(\{a\}\)).
	\label{example:stable-supported}
\end{example}

\begin{figure}[ht]
	\centering
	\begin{subfigure}[b]{0.45\textwidth}
		\centering
		\begin{tikzpicture}[node distance=1cm and 1cm, every node/.style={scale=1.0}]
		\node[circle, draw] (a) [] {$a$};
		\node[circle, draw] (b) [right=of a, xshift=0cm] {$b$};
		\node[circle, draw] (c) [right=of b, xshift=0cm] {$c$};
		
		\draw[->] (b) edge [bend right=25] node [midway, above, fill=white] {$\oplus$} (a);
		\draw[->] (b) edge [bend left=25] node [midway, above, fill=white] {$\ominus$} (a);
		
		\draw[->] (c) edge [bend left=25] node [midway, above, fill=white] {$\oplus$} (b);
		\draw[->] (b) edge [bend left=25] node [midway, above, fill=white] {$\oplus$} (c);
		\end{tikzpicture}
		\caption{}
	\end{subfigure}%
	\begin{subfigure}[b]{0.45\textwidth}
		\centering
		\begin{tikzpicture}[node distance=1cm and 1cm, every node/.style={scale=1.0}]
		\node[circle, draw] (a) [] {$a$};
		\node[circle, draw] (b) [right=of a, xshift=0cm] {$b$};
		\node[circle, draw] (c) [right=of b, xshift=0cm] {$c$};
		
		\draw[->] (c) edge [bend left=25] node [midway, above, fill=white] {$\oplus$} (b);
		\draw[->] (b) edge [bend left=25] node [midway, above, fill=white] {$\oplus$} (c);
		\end{tikzpicture}
		\caption{}
	\end{subfigure}%
	\caption{(a) \(\dg{P}\). (b) \(\ig{f}\).}
	\label{fig:dg-ig-example}
\end{figure}

\subsection{Cycles and Stable Models}

We start with revisiting one well-known result in ASP (see Proposition~\ref{proposition:acyclic-program}).
We provide a new proof for it, which relies on the relation between an LP and a BN.

\begin{proposition}
	Let \(P\) be an LP.
	If \(\dg{P}\) has no cycle, then \(P\) has exactly one stable model.
	\label{proposition:acyclic-program}
\end{proposition}
\begin{proof}
	Let \(f\) be the encoded BN of \(P\).
	Since \(\pdg{P} \subseteq \dg{P}\), \(\pdg{P}\) has no cycle, then \(\sm{P} = \supp{P}\) by Theorem~\ref{theo:fages}.
	Since \(\ig{f} \subseteq \dg{P}\), \(\ig{f}\) has no cycle.
	By Theorem 1 of~\cite{richard2019positive}, \(f\) has a unique fixed point, thus \(P\) has a unique stable model.
\end{proof}

Proposition~\ref{proposition:acyclic-program} is the consequence of Theorem 2.5(iv) of~\cite{apt1991acyclic} and Fages' theorem.
Theorem 2.5(iv) of~\cite{apt1991acyclic} shows that if \(\dg{P}\) has no cycle, then \(\cf{P}\) has a unique model, which is also the unique stable model of \(P\) by Fages' theorem.
Indeed, the BN-based proof that we provide here is quite simpler than the above proof.
We then present two newly small results considering LPs with very special structures (see Propositions~\ref{proposition:positive-cycle} and~\ref{proposition:negative-cycle}).
See Example~\ref{example:pos-neg-cycle} for illustration.

\begin{proposition}
	Let \(P\) be an LP.
	Suppose that \(\dg{P}\) is a positive cycle.
	If \(\dg{P}\) has a negative arc, then \(P\) has exactly two stable models.
	Otherwise, \(P\) has exactly one stable model.
	\label{proposition:positive-cycle}
\end{proposition}
\begin{proof}
	Let \(f\) be the BN encoding of \(P\).
	Since \(\dg{P}\) is a positive cycle, the in-degree of each vertex is one, leading to every \(f_v\) is constant-free.
	Hence, \(\ig{f}\) is also a positive cycle.
	By~\cite{remy2003description}, \(f\) has exactly two fixed points.
	If \(\dg{P}\) has a negative arc, then \(\pdg{P}\) has no cycle.
	It follows that \(\supp{P} = \sm{P}\) by Theorem~\ref{theo:fages}.
	Hence, \(P\) has exactly two stable models.
	If \(\dg{P}\) has no negative arc, then \(P\) is positive.
	Hence, it has exactly one stable model.
\end{proof}

In the proof of Proposition~\ref{proposition:positive-cycle}, if \(P\) has exactly two stable models (say \(A\) and \(B\)), then \(A \cap B = \emptyset\) and \(A \cup B = \atom{P}\), since \(A\) and \(B\) are also the two fixed points of \(f\) where \(A_v \neq B_v, \forall v \in \var{f}\) by~\cite{remy2003description}.
Otherwise, the unique stable model of \(P\) is \(\emptyset\).

\begin{proposition}
	Let \(P\) be an LP.
	Suppose that \(\dg{P}\) is a negative cycle.
	Then \(P\) has no stable model.
	\label{proposition:negative-cycle}
\end{proposition}
\begin{proof}
	Let \(f\) be the BN encoding of \(P\).
	Since \(\dg{P}\) is a negative cycle, the in-degree of each vertex is one, leading to every \(f_v\) is constant-free.
	Hence, \(\ig{f}\) is also a negative cycle.
	By~\cite{remy2003description}, \(f\) has no fixed point.
	It follows that \(P\) has no supported model.
	Since a stable model is also a supported model, \(P\) has no stable model.
\end{proof}

\begin{example} Let \(P_1\) be the LP: \(a \leftarrow \dng{b}; b \leftarrow \dng{c}; c \leftarrow a\).
	Let \(P_2\) be the LP: \(a \leftarrow b; b \leftarrow \dng{c}; c \leftarrow a\).
	Figures~\ref{fig:pos-neg-cycle}(a) and~\ref{fig:pos-neg-cycle}(b) show the DGs of \(P_1\) and \(P_2\), respectively.
	\(\dg{P_1}\) is a positive cycle and it has two negative arcs.
	\(P_1\) has two stable models (\(A_1 = \{a, c\}, A_2 = \{b\}\)), which is consistent with Proposition~\ref{proposition:positive-cycle}.
	We also see that \(A_1 \cap A_2 = \emptyset\) and \(A_1 \cup A_2 = \{a, b, c\} = \atom{P_1}\).
	\(\dg{P_2}\) is a negative cycle and \(P_2\) has no stable model, which is consistent with Proposition~\ref{proposition:negative-cycle}.
	\label{example:pos-neg-cycle}
\end{example}

\begin{figure}[ht]
	\centering
	\begin{subfigure}[b]{0.45\textwidth}
		\centering
		\begin{tikzpicture}[node distance=1cm and 1cm, every node/.style={scale=1.0}]
		\node[circle, draw] (a) [] {$a$};
		\node[circle, draw] (b) [right=of a, xshift=0cm] {$b$};
		\node[circle, draw] (c) [right=of b, xshift=0cm] {$c$};
		
		\draw[->] (b) edge [] node [midway, below, fill=white] {$\ominus$} (a);
		\draw[->] (c) edge [] node [midway, below, fill=white] {$\ominus$} (b);
		\draw[->] (a) edge [bend left=25] node [midway, above, fill=white] {$\oplus$} (c);
		\end{tikzpicture}
		\caption{}
	\end{subfigure}%
	\begin{subfigure}[b]{0.45\textwidth}
		\centering
		\begin{tikzpicture}[node distance=1cm and 1cm, every node/.style={scale=1.0}]
		\node[circle, draw] (a) [] {$a$};
		\node[circle, draw] (b) [right=of a, xshift=0cm] {$b$};
		\node[circle, draw] (c) [right=of b, xshift=0cm] {$c$};
		
		\draw[->] (b) edge [] node [midway, below, fill=white] {$\oplus$} (a);
		\draw[->] (c) edge [] node [midway, below, fill=white] {$\ominus$} (b);
		\draw[->] (a) edge [bend left=25] node [midway, above, fill=white] {$\oplus$} (c);
		\end{tikzpicture}
		\caption{}
	\end{subfigure}%
	\caption{(a) \(\dg{P_1}\). (b) \(\dg{P_2}\).}
	\label{fig:pos-neg-cycle}
\end{figure}
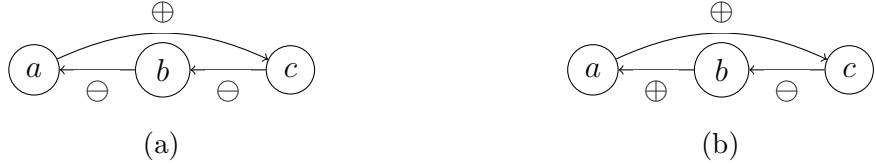

Next, we investigate more deeply the effect of positive and negative cycles of the DG of \(P\) on its stable models.
Recall that a Positive Feedback Vertex Set (PFVS) of a signed directed graph \(G\) is a set of vertices intersecting all positive cycles of \(G\)~\cite{richard2019positive}.

\begin{lemma}
	Let \(P\) be an LP.
	If \(P\) has two distinct stable models (say \(A_1\) and \(A_2\)), then \(\dg{P}\) has a positive cycle \(C^+\) such that for every \(v \in C^+\), either \(v \in A_1\) or \(v \in A_2\).
	\label{lemma:distinct-stable-model-positive-cycle}
\end{lemma}
\begin{proof}
	Let \(f\) be the BN encoding of \(P\).
	\(A_1\) and \(A_2\) are also distinct supported models of \(P\).
	By Theorem~\ref{theorem:supported-fixed-point}, \(A_1\) and \(A_2\) are distinct fixed points of \(f\).
	By Lemma 1\footnote{If \(f\) has two fixed points \(x\) and \(y\), then \(\ig{f}\) has a positive cycle \(C^+\) such that \(x_v \neq y_v, \forall v \in C^+\).} of~\cite{richard2019positive}, \(\ig{f}\) has a positive cycle \(C^+\) such that for every \(v \in C^+\), \(A_1(v) \neq A_2(v)\), i.e., either \(v \in A_1\) or \(v \in A_2\) because \(A_1(v), A_2(v) \in \mathbb{B}\).
	\(C^+\) is also a positive cycle of \(\dg{P}\) because \(\ig{f}\) is a subgraph of \(\dg{P}\).
\end{proof}

\begin{lemma}
	If a sign directed graph \(G\) is strongly connected and has no negative cycle or has no positive cycle, then \(G\) is sign-definite.
	\label{lemma:sign-definite-graph-scc}
\end{lemma}
\begin{proof}
	We first prove that each arc of \(G\) belongs to a cycle in \(G\) (*).
	Taken an arbitrary arc \((uv, \epsilon)\) in \(G\).
	Since \(G\) is strongly connected, there is a directed path from \(v\) to \(u\).
	By adding \((uv, \epsilon)\) to this path, we obtain a cycle.
	
	Assume that \(G\) is not sign-definite.
	Then there are two arcs: \((uv, \oplus)\) and \((uv, \ominus)\).
	By (*), \((uv, \oplus)\) (resp. \((uv, \ominus)\)) belongs to a cycle in \(G\) (say \(C\)).
	\(C\) is a positive (resp. negative) cycle because \(G\) has no negative (resp. positive) cycle.
	Then \((C - (uv, \oplus)) + (uv, \ominus)\) (resp. \((C - (uv, \ominus)) + (uv, \oplus)\)) is a negative (resp. positive) cycle in \(G\).
	This implies a contradiction.
	Hence, \(G\) is sign-definite.
\end{proof}

\begin{theorem}
	Let \(P\) be an LP.
	If \(\dg{P}\) has no positive cycle, then \(P\) has at most one stable model.
	In addition, if \(P\) has no fact and every atom appears in the head of a rule in \(P\), then \(P\) has no stable model.
	\label{theorem:no-positive-cycle}
\end{theorem}
\begin{proof}
	Assume that \(P\) has at least two stable models.
	Then \(\dg{P}\) has a positive cycle by Lemma~\ref{lemma:distinct-stable-model-positive-cycle}, which is a contradiction.
	Hence, \(P\) has at most one stable model.
	
	Since \(P\) no fact and every atom appears in the head of a rule in \(P\), each vertex of \(\dg{P}\) has at least one input vertex.
	Let \(G^*\) be a strongly connected subgraph of \(\dg{P}\) such that there is no arc from a vertex outside \(G^*\) to a vertex in \(G^*\).
	Since \(\dg{P}\) has no positive cycle, then \(G^*\) also has no positive cycle.
	By Lemma~\ref{lemma:sign-definite-graph-scc}, \(G^*\) is sign-definite.
	Let \(f\) be the BN encoding of \(P\).
	Since there is no arc from a vertex outside \(G^*\) to a vertex in \(G^*\), there is a BN \(f^*\) on the set of variables that are vertices of \(G^*\) and \(f^*_v = f_v, \forall v \in \var{f^*}\).
	
	Since each vertex of \(\dg{P}\) has at least one input vertex, each vertex of \(G^*\) also has at least one input vertex.
	Let \(v\) be an arbitrary vertex in \(G^*\).
	Since \(G^*\) is sign-definite, the sets of positive and negative input vertices of \(v\) are disjoint.
	By assigning 0 to all positive input vertices of \(v\) in \(G^*\) and 1 to all negative input vertices of \(v\) in \(G^*\), we have \(f^*_v = 0\).
	Symmetrically, assigning 1 to all positive input vertices of \(v\) in \(G^*\) and 0 to all negative input vertices of \(v\) in \(G^*\), we have \(f^*_v = 1\).
	Hence, \(f^*_v\) cannot be constant.
	It implies that the minimum in-degree of \(ig(f^*)\) is at least one.
	Clearly, \(ig(f^*)\) is a subgraph of \(G^*\), leading to \(ig(f^*)\) has no positive cycle.
	Then \(f^*\) has no fixed point by Theorem 2\footnote{If \(\ig{f}\) has no positive cycle and its minimum in-degree is at least one, then \(f\) has no fixed point.} of~\cite{aracena2004positive}.
	Let \(\gamma(f)\) and \(\gamma(f^*)\) be the SAT formulas characterizing fixed points of \(f\) and \(f^*\), respectively.
	Since \(f^*\) has no fixed point, \(\gamma(f^*)\) has no model.
	Clearly, \(\gamma(f^*)\) is part of \(\gamma(f)\), thus \(\gamma(f)\) has no model.
	Hence, \(f\) has no fixed point.
	It follows that \(P\) has no supported model.
	Since a stable model is also a supported model, \(P\) has no stable model.
\end{proof}

\begin{proposition}
	Let \(P\) be an LP.
	For any PFVS \(U^{+}\)  of \(\dg{P}\), we have that \(|\sm{P}| \leq 2^{|U^{+}|}\).
	\label{proposition:ubound-pfvs}
\end{proposition}
\begin{proof}
	Let \(f\) be the BN encoding of \(P\).
	Since \(U^{+}\) is a PFVS of \(\dg{P}\), \(\dg{P} - U^{+}\) has no positive cycle.
	Since \(\ig{f} \subseteq \dg{P}\), \(\ig{f} - U^{+}\) also has no positive cycle.
	Therefore, \(U^{+}\) is a PFVS of \(\ig{f}\).
	By Theorem 9 of~\cite{aracena2008maximum}, \(|\fix{f}| \leq 2^{|U^{+}|}\).
	Since \(|\sm{P}| \leq |\fix{f}|\), we can conclude that \(|\sm{P}| \leq 2^{|U^{+}|}\).
\end{proof}

Note that Proposition~\ref{proposition:ubound-pfvs} can be directly proved by using Theorem~\ref{theorem:no-positive-cycle}.
The main idea is that in any stable model of \(P\), each atom in \(U^{+}\) can be either true or false.
For each assignment of atoms in \(U^{+}\), we get a new LP whose DG has no positive cycle, then it has at most one stable model by Theorem~\ref{theorem:no-positive-cycle}.
There are \(2^{|U^{+}|}\) possible assignments of atoms in \(U^{+}\), hence we can deduce \(|\sm{P}| \leq 2^{|U^{+}|}\).

It is well-known in the BN field that if \(\ig{f}\) has no negative cycle, then \(f\) has at least one fixed point (Theorem 6 of~\cite{aracena2008maximum}).
Naturally, it is interesting to question whether an LP whose DG has no negative cycle has at least one stable model or not.
Fages showed a counterexample for the case of infinite logic programs~\cite{cois1994consistency}, but the case of finite logic programs is still open.
We answer this question by Theorem~\ref{theorem:neg-model-existence}.

To prove Theorem~\ref{theorem:neg-model-existence}, we use the fixpoint semantics of logic programs~\cite{dung1989fixpoint}.
To be self-contained, we briefly recall the definition of the least fixpoint of a logic program.
A \emph{quasi-interpretation} is a set possibly infinite of rules of the form \(p \leftarrow \dng{p_1}, \dots, \dng{p_k}\) where \(k \geq 0\) and \(p, p_1, \dots, p_k \in \atom{P}\).
Let \(r\) be the rule \(p \leftarrow \dng{p_1}, \dots, \dng{p_k}, q_1, \dots, q_j\) and let \(r_i\) be rules \(q_i \leftarrow \dng{q^1_i}, \dots, \dng{q^{l_i}_i}\) where \(1 \leq i \leq j\) and \(l_i \geq 0\).
Then \(T_{r}(\{r_1, \dots, r_j\})\) is the rule \(p \leftarrow \dng{p_1}, \dots, \dng{p_k}, \dng{q_1^1}, \dots, \dng{q_1^{l_1}}, \dots, \dng{q_j^1}, \dots, \dng{q_j^{l_j}}.\) We now introduce the transformation \(T_P\) on quasi-interpretations: \(T_P(Q) = \{T_r(\{r_1, \dots, r_j\}) | r \in P, r_i \in Q, 1 \leq i \leq j\}.\) Let \(\text{lfp}_i = T_P^i(\emptyset) = T_P(T_P(\dots T_P(\emptyset)))\), then \(\text{lfp} = \bigcup_{i \geq 1}\text{lfp}_i\) is the \emph{least fixpoint} of \(P\).

\begin{lemma}
	\label{lemma:fixpoint-neg}
	Let \(P\) be an LP and \(P'\) be its least fixpoint.
	If \(\dg{P}\) has no negative cycle, then \(\dg{P'}\) has no negative cycle.
\end{lemma}
\begin{proof}
	By Lemma 5.3 of~\cite{cois1994consistency}, if vertex \(a\) has a negative path to vertex \(b\) in \(\dg{P'}\), then \(a\) has also a negative path to \(b\) in \(\dg{P}\).
	Note that the set of vertices of \(\dg{P'}\) is a subset of that of \(\dg{P}\).
	It follows that if there is a negative cycle in \(\dg{P'}\), then there is also a negative cycle in \(\dg{P}\).
	Hence, if \(\dg{P}\) has no negative cycle, then \(\dg{P'}\) has no negative cycle.
\end{proof}

\begin{theorem}
	\label{theorem:neg-model-existence}
	Let \(P\) be an LP.
	If \(\dg{P}\) has no negative cycle, then \(P\) has at least one stable model.
\end{theorem}
\begin{proof}
	Let \(P'\) be the least fixpoint of \(P\).
	By Lemma~\ref{lemma:fixpoint-neg}, \(\dg{P'}\) has no negative cycle.
	Let \(f'\) be the encoded BN of \(P'\).
	Since \(\ig{f'} \subseteq \dg{P'}\), \(\ig{f'}\) also has no negative cycle.
	Therefore, \(f'\) has at least one fixed point by Theorem 6 of~\cite{aracena2008maximum}.
	By Theorem~\ref{theorem:supported-fixed-point}, \(\supp{P'} = \fix{f'}\).
	It is known that \(\sm{P} = \supp{P'}\) (Theorem 4 of~\cite{dung1989fixpoint}).
	This implies that \(\sm{P} = \fix{f'}\).
	Hence, \(P\) has at least one stable model.
\end{proof}

The technique used to prove Theorem~\ref{theorem:neg-model-existence} suggests a potential way to explore more theoretical results: picking up a structural property of \(\dg{P}\), seeing what it becomes in the DG of the least fixpoint of \(P\), and checking how fixed points of the encoded BN of the least fixpoint behave under the new property.
For illustration, consider the class of locally stratified LPs.
\(P\) is \emph{locally stratified} if every cycle of \(\dg{P}\) contains no negative arc~\cite{gelfond1988stable}.
Considering this property, we can prove that the DG of the least fixpoint of \(P\) has no cycle.
In this case, the least fixpoint has exactly one stable model, leading to \(P\) has so.
See the detailed proof in Theorem~\ref{theorem:locally-stratified-lp}.
This provides an alternative (maybe simpler) proof for the well-known result stating that a locally stratified LP has a unique stable model (Theorem 6.2.7 of~\cite{hitzler2011mathematical}).

\begin{theorem}
	\label{theorem:locally-stratified-lp}
	Let \(P\) be an LP.
	If \(P\) is locally stratified, then \(P\) has exactly one stable model.
\end{theorem}
\begin{proof}
	Let \(P'\) be the least fixpoint of \(P\).
	Recall that \(P' = \text{lfp} = \bigcup_{i \geq 1}\text{lfp}_i\).
	We show that \(\dg{\text{lfp}_i}\) has no cycle by induction on \(i\).
	
	The base case is trivial (i.e., \(i = 0\)).
	Consider an arbitrary rule \(r \in \text{lfp}_{i + 1} \setminus \text{lfp}_i\).
	Let \(b = \head{r}\) and \(a \in \nbody{r}\).
	By definition, there exists a rule \(r' \in P\) such that \(\head{r'} = b\) and \(a \in \nbody{r''}\) with \(r'' \in \text{lfp}_i\) and \(\head{r''} \in \pbody{r'}\).
	Let \(c = \head{r''}\).
	We use \(b \leq_{\oplus} a\) (resp. \(b \leq_{\ominus} a\)) to denote that there is a positive (resp. negative) path from \(b\) to \(a\).
	We have that \(a \leq_{\ominus} b\) in \(\dg{\text{lfp}_{i + 1}}\).
	Note that \(\text{lfp}_i \subseteq \text{lfp}_{i + 1}\) by definition.
	If \(b \leq_{\oplus} a\) in \(\dg{\text{lfp}_i}\), then \(b \leq_{\ominus} a\) in \(\dg{\text{lfp}_{i + 1}}\), leading to \(b \leq_{\ominus} b\) in \(\dg{\text{lfp}_{i + 1}}\) (also in \(\dg{P'}\)), which is a contradiction because \(\dg{P'}\) has no negative cycle by Lemma 5.3 of~\cite{cois1994consistency}.
	If \(b \leq_{\ominus} a\) in \(\dg{\text{lfp}_i}\), then \(b \leq_{\ominus} a\) in \(\dg{P'}\), leading to \(b \leq_{\ominus} a\) in \(\dg{P}\) by Lemma 5.3 of~\cite{cois1994consistency}.
	Similarly, \(a \leq_{\ominus} c\) in \(\dg{P}\), since \(a \leq_{\ominus} c\) in \(\dg{\text{lfp}_i}\).
	We use \(b \leq_{0} a\) to denote that there is a path of only positive arcs from \(b\) to \(a\).
	We have \(c \leq_{0} b\) in \(\dg{P}\).
	Then there is a cycle containing \(b\) and also negative arcs in \(\dg{P}\), which is a contradiction because \(P\) is locally stratified.
	Hence, there is no path from \(b\) to \(a\) in \(\dg{\text{lfp}_i}\).
	Since \(\dg{\text{lfp}_i}\) has no cycle by the induction hypothesis and adding the rule \(r\) to \(\text{lfp}_i\) does not introduce new cycles, \(\dg{\text{lfp}_{i + 1}}\) has no cycle.
	
	It follows that \(\dg{P'}\) has no cycle, then \(P'\) has exactly one stable model by Proposition~\ref{proposition:acyclic-program}.
	Since \(\sm{P} = \sm{P'}\) by Theorem 4 of~\cite{dung1989fixpoint}, \(P\) also has exactly one stable model.
\end{proof}

Finally, inspired by Theorem~\ref{theorem:neg-model-existence}, we explore an interesting result shown in Theorem~\ref{theorem:one-scc-two-stable-models}, which is actually the generalization of Proposition~\ref{proposition:positive-cycle}.

\begin{theorem}
	Let \(P\) be an LP.
	Suppose that \(\dg{P}\) is strongly connected, has at least one arc, and has no negative cycle.
	If \(\pdg{P}\) has no cycle, then \(P\) has two stable models \(A\) and \(B\) such that \(\forall v \in \atom{P}\), either \(v \in A\) or \(v \in B\).
	In addition, \(A\) and \(B\) can be computed in polynomial time.
	\label{theorem:one-scc-two-stable-models}
\end{theorem}
\begin{proof}
	Let \(f\) be the BN encoding of \(P\).
	Since \(\pdg{P}\) has no cycle, \(\sm{P} = \supp{P} = \fix{f}\) by Theorems~\ref{theo:fages} and~\ref{theorem:supported-fixed-point}.
	We show that \(f\) has two fixed points that are complementary.
	
	Since \(\dg{P}\) is strongly connected and has no negative cycle, it is sign-definite by Lemma~\ref{lemma:sign-definite-graph-scc}.
	It implies that \(\ig{f}\) is also sign-definite because \(\ig{f} \subseteq \dg{P}\).
	\(\dg{P}\) has the minimum in-degree of at least one because it is strongly connected and has at least one arc.
	By using the deduction similar to that in Theorem~\ref{theorem:no-positive-cycle}, we have that \(f\) has no constant function.
	It is known that when \(\dg{P}\) is strongly connected and has no negative cycle, its set of vertices can be divided into two equivalence classes (say \(S^+\) and \(S^-\)) such that any two vertices in \(S^+\) (resp. \(S^-\)) are connected by either no arc or a positive arc, and there is either no arc or a negative arc between two vertices in \(S^+\) and \(S^-\) (Theorem 1 of~\cite{akutsu2012singleton}).
	Since \(\ig{f} \subseteq \dg{P}\) and \(\ig{f}\) has the same set of vertices with \(\dg{P}\), \(S^+\) and \(S^-\) are still such two equivalence classes in \(\ig{f}\).
	
	Let \(x\) be a state defined as: \(x_i = 1\) if \(i \in S^+\) and \(x_i = 0\) if \(i \in S^-\).
	Consider a given node \(j\).
	If \(x_j = 0\), by the above result, for all \(i \in \atom{P}\) such that \(\ig{f}\) has a positive arc from \(i\) to \(j\), \(x_i = 1\), and for all \(i \in \atom{P}\) such that \(\ig{f}\) has a negative arc from \(i\) to \(j\), \(x_i = 0\).
	Since \(f_j\) cannot be constant, \(f_j(x) = 1\).
	Analogously, if \(x_j = 1\), then \(f_j(x) = 1\), implying that \(x\) is a fixed point of \(f\).
	By using the similar deduction, we can conclude that \(\overline{x}\) is also a fixed point of \(f\) where \(\overline{x}_i = 1 - x_i, \forall i \in \var{f}\).
	Let \(A\) and \(B\) are two models of \(P\) corresponding to \(x\) and \(\overline{x}\).
	Clearly, \(A\) and \(B\) are stable models of \(P\).
	We have that \(\forall v \in \atom{P}\), either \(v \in A\) or \(v \in B\).
	In addition, since \(S^+\) and \(S^-\) can be computed in polynomial time, \(A\) and \(B\) can be computed in polynomial time.
\end{proof}

\begin{example} Consider an LP \(P\): \(a \leftarrow \dng{b}, b \leftarrow \dng{a}, b \leftarrow \dng{c}, c \leftarrow \dng{b}\).
	The encoded BN \(f\) of \(P\) is: \(f_a = \neg b, f_b = \neg a \lor \neg c, f_c = \neg b\).
	\dg{P} is strongly connected, has at least one arc, and has no negative cycle.
	Since \(\pdg{P}\) has no cycle, \(P\) has two stable models.
	Actually, \(P\) has two stable models: \(A = \{b\}\) and \(B = \{a, c\}\).
	We easily see that \(A \cap B = \emptyset\) and \(A \cup B = \atom{P}\).
	\label{example:two-stable-models}
\end{example}

\section{Discussion}
\label{sec:Discussion}

\subsection{Computation of stable models}
Theorem~\ref{theorem:supported-fixed-point} and the subsequent results immediately suggest that we can compute stable models of an LP \(P\) by using fixed points of its encoded BN \(f\).
Recall that a fixed point of \(f\) may be not a stable model of \(P\), but checking whether a fixed point is a stable model or not can be done in linear time~\cite{eiter2009answer}.
Notably, there is a rich history for computing fixed points in BNs: SAT-based methods~\cite{mori2022attractor}, ILP-based methods~\cite{mori2022attractor}, and structure-based methods~\cite{aracena2021finding}.
We can first check if \(\pdg{P}\) does not contain any cycle.
If so, \(\sm{P} = \fix{f}\) and we can apply directly some efficient methods for fixed point computation.
Otherwise, we need to check if each fixed point is a stable model or not.
This approach would be efficient if \(|\fix{f} \setminus \sm{P}|\) is not too large.
Note however that we can estimate \(|\fix{f}|\) in prior by using some known upper bounds of the number of fixed points of a BN (e.g., \(2^{|U^{+}|}\) where \(U^{+}\) is a PFVS of \(\dg{P}\), which can be efficiently computed by using some approximation methods~\cite{richard2019positive} because \(U^{+}\) is unnecessarily minimum).
If the chosen bound is too large, we can apply some transformations (e.g.,~\cite{janhunen2011compact}) to \(P\) to get a new LP \(P'\) such that \(\sm{P} = \sm{P'} = \fix{f'}\) where \(f'\) is the encoded BN of \(P'\).
Such a transformation might introduce more atoms, but it seems that there are less changes in cycle structures (e.g., the size of the minimum PFVS seems to still retain~\cite{janhunen2011compact}).
In this case, structure-based methods for fixed point computation, especially the PFVS-based method~\cite{aracena2021finding} whose complexity is \(O(2^{|U^{+}|} \times n^{2 + k})\) where the Boolean functions of the BN can be evaluated in time \(O(n^k)\), would be very helpful in complementary to standard ASP solvers such as Cmodels, Smodels, and clasp.

\subsection{Program correction} 
Theorem~\ref{theorem:neg-model-existence} has the potential to be useful in some other problems in ASP.
It is related to the program correction problem~\cite{janota2014minimal}, where we intend to modify an inconsistent LP \(P_{incon}\) (i.e., having no stable model) to make it consistent (i.e., having at least one stable model).
For example, by making some modification operators (e.g., removing rules, flipping signs of literals) such that \(\dg{P_{incon}}\) has no negative cycle, we can obtain the consistency for \(P_{incon}\).
These modifications might be not minimal, but they might provide a good upper bound for further searching.
In the opposite direction, we might want to modify a consistent LP \(P_{con}\) to make it inconsistent.
By Theorem~\ref{theorem:no-positive-cycle}, we can make modifications such that \(\dg{P_{con}}\) has no positive cycle.
Such a problem might be useful in program verification where the validation of a property can be encoded as the unsatisfiability of an SAT formula~\cite{luo2022proving} or the inconsistency of an LP.
We believe that exploiting Theorems~\ref{theorem:no-positive-cycle} and~\ref{theorem:neg-model-existence} will give more benefits to the program correction problem, which we plan to study deeply in future.

\subsection{Interplay between positive and negative cycles} 
The results presented in this paper show that positive and negative cycles are key structures to understand the relations between stable models and the DG of an LP.
However, they use only information on either positive cycles or negative cycles.
It is then natural to think that by using both kinds of cycles simultaneously, we can obtain more improved results.
Indeed, this approach for BNs has been thoroughly investigated with many insightful results~\cite{richard2018fixed,richard2019positive,richard2023attractor}.
For example, it has been shown that if the IG of a BN has negative cycles but these cycles are isolated by positive cycles, then the BN behaves as in the absence of negative cycles, i.e., it has at least one fixed point~\cite{richard2018fixed}.
Adapting these results for BNs to LPs seems to be promising as we are not aware of any similar studies in LPs.
However, this direction is non-trivial because the mix between positive and negative cycles might hinder the techniques that we propose in this paper.
Other techniques, maybe still being under the umbrella of the connection between LPs and BNs, might be needed.

\section{Conclusion and Future Work}
\label{sec:Conclusion}

Static analysis of ASP is important and has been proved very useful in both theory and practice of ASP.
In this work, we for the first time bridged between ASP and BNs, and used this connection to study the static analysis of ASP at that depth.
Specifically, we stated and proved several relations between positive and negative cycles in the DG of an LP and its stable models.
The most important obtained results include 1) the existence and the non-existence of some stable model under the non-existence of negative cycles and the non-existence of positive cycles, respectively, 2) an upper bound of the number of stable models based on PFVSs, and 3) the connection and techniques we established, which provide a powerful and unified framework for exploring and proving more new theoretical results in ASP.
In particular, we also demonstrated that the obtained results have the potential to be useful for solving two important problems in the field of ASP.

In the future, on the one hand, we shall explore more theoretical results following up the approach we established in this paper.
First, we shall continue to discover more results by considering how a structural property of the DG of an LP transforms in its least fixpoint.
Second, we plan to investigate how the interplay between positive and negative cycles of the DG affects the set of stable models.
Third, we shall generalize the obtained results for disjunctive logic programs, extended logic programs (i.e., containing strong negations), or logic programs with aggregates, constraints, and choice rules.
On the other hand, we plan to investigate thoroughly the applications of the obtained results.
First, we shall implement the proposed approach for stable model computation discussed in the previous section.
Then, we need to thoroughly test the efficiency of this approach on a large set of real-world problem instances, maybe obtained from the ASP competitions.
Second, we shall develop new methods for program correction that rely on the removal of positive and negative cycles, and then thoroughly test their efficiency.


\section*{Acknowledgments}

This work was supported by Institut Carnot STAR, Marseille, France.

\bibliographystyle{plain}
\bibliography{main}
\end{document}